\documentclass[journal]{IEEEtran}
\usepackage{graphicx} 
\usepackage{lipsum}
\usepackage[T1]{fontenc}
\usepackage{amsmath}

\usepackage{amssymb}
\usepackage{amsthm}
\usepackage{stackengine}
\usepackage{lipsum}
\usepackage{multicol}

\usepackage[usenames, dvipsnames]{color}
\usepackage{caption}
\usepackage{subcaption}
\newtheorem{theorem}{Theorem}

\newtheorem{remark}{Remark} 

\newtheorem{assum}{Assumption}
\newtheorem{corollary}{Corollary}
\newtheorem{propos}{Proposition}

\usepackage{lipsum}
\usepackage{mathtools}
\usepackage{cuted}
\usepackage{algorithm,algorithmic}
\usepackage{epstopdf}
\usepackage{cite}
\usepackage{amsmath,amssymb,amsfonts}
\usepackage{algorithmic}
\usepackage{textcomp}
\begin{document}
\title{Robust Lattice-based Motion Planning}
\author{Abhishek Dhar, Carl Hynén, 
Johan Löfberg and Daniel Axehill
\thanks{This work is supported by ELLIT.}
\thanks{The authors are with the Division of Automatic Control, Department of Electrical Engineering, Linköping University, Sweden. e-mail:
       ({\tt\small abhishek.dharr@gmail.com,carl.hynen@liu.se, johan.lofberg@liu.se,daniel.axehill@liu.se})}}

\maketitle

\begin{abstract}
This paper proposes a robust lattice-based motion-planning algorithm for nonlinear systems affected by a bounded disturbance. The proposed motion planner utilizes the nominal disturbance-free system model to generate motion primitives, which are associated with fixed-size tubes. These tubes are characterized through designing a feedback controller, that guarantees boundedness of the errors occurring due to mismatch between the disturbed nonlinear system and the nominal system. The motion planner then sequentially implements the tube-based motion primitives while solving an online graph-search problem. The objective of the graph-search problem is to connect the initial state to the  final state, through sampled states in a suitably discretized state space, such that the tubes do not pass through any unsafe states (representing obstacles) appearing during runtime. The proposed strategy is implemented on an Euler-Lagrange based ship model which is affected by significant wind disturbance. It is shown that the uncertain system trajectories always stay within a suitably constructed tube around the nominal trajectory and terminate within a region around the final state, whose size is dictated by the size of the tube.    
\end{abstract}
\vspace{-0.in}
\section{Introduction}
\IEEEPARstart{T}{he} problem of motion-planning deals with generating a feasible trajectory, which connects an initial state of a system to a desired final state, while avoiding unsafe states, which represent physical constraints of the system as well as obstacles in the environment. Several motion-planning strategies are proposed in literature, which can be broadly classified into combinatorial strategies \cite{krogh,rimon,lingelbach} and sampling-based strategies \cite{karaman,lavalle}. The combinatorial strategies such as the potential field approach \cite{krogh}, the navigation function based approach \cite{rimon}, cell decomposition methods \cite{lingelbach}, etc. aim at utilizing the knowledge of the environment and to create a feasible path from an initial state to a final state either by utilizing a continuous model of the feasible environment or by systematically combining cells, obtained by decomposing the environment. On the other hand, the sampling-based strategies, most of which are based on the standard Rapidly exploring Random Tree (RRT) and RRT* algorithms \cite{karaman,lavalle}, aim at random explorations in the concerned state space and finding sample states, which can be feasibly connected. However, these algorithms are not inherently robust to the uncertainties which appear in practise due to various factors, such as incomplete knowledge of the model, unknown obstacles, exogenous disturbances, etc. It is essential for a motion-planning algorithm to generate robust trajectories which are feasible even under the effect of uncertainties, to ensure safety of the system as well as of the environment.\par
To guarantee safe motion in cluttered environment in presence of uncertainties, various efficient robust motion-planning algorithms have been proposed in literature. The strategies proposed in \cite{luders,lindemann,hakobyan,blackmore,bahreinian,danielson} deal with the motion-planning problem for linear systems affected by uncertainties. The results in \cite{luders,lindemann} propose robust RRT strategies, which extend the standard RRT algorithm to the cases of linear systems having additive noise and uncertainty in state estimation. The approaches in \cite{hakobyan,bahreinian} guarantee safe motion for systems with linear models, in environments with unknown obstacles. An efficient invariant-set based robust motion planner is proposed in \cite{danielson}, which handles linear systems with parametric uncertainty as well as bounded additive disturbances. Although these approaches are efficient robust motion-planning solutions, they are not suitable to handle systems with nonlinear models, which is mostly the case in practise. The approaches in \cite{tsukamoto,manjunath,majumdar,gurgen,singh} address the robust motion-planning problem for nonlinear systems with uncertainties. A learning-based robust motion planner is proposed in \cite{tsukamoto}, which utilizes contraction theory to generate a safety certificate for trajectories of a nonlinear system affected by additive disturbances. However, to guarantee performance, the strategy in \cite{tsukamoto} requires complete knowledge of the obstacles in the environment, which is restrictive for practical implementations. A control barrier function based robust RRT strategy is proposed for uncertain nonlinear systems in \cite{manjunath} and the performance is validated through numerical simulations.  The strategies in \cite{majumdar,gurgen} generate funnel libraries to handle uncertainties in robust motion-planning. These strategies can efficiently handle nonlinear systems under the effect of uncertainties, such as external disturbances and obstacles which appear during runtime. However, as the funnels vary in shape and size, additional efforts/conditions are required to maintain feasible composition of the funnels during online implementation. The drawback associated with the funnel library based motion-planning is circumvented in the strategy in \cite{singh}, which generates a fixed-size tube along a planned nominal trajectory, to capture the motion of the nonlinear system under uncertainty. However, the approach in \cite{singh} depends upon an external motion planner to compute the nominal trajectory and it requires replanning during runtime if the currently planned path becomes infeasible due to the appearance of any previously unforeseen obstacle. \par
This paper proposes a novel robust lattice-based motion-planning strategy for nonlinear systems which is affected by bounded a exogenous disturbance. The lattice based motion planner \cite{bergman,bergman1,ljungqvist} is an efficient sampling-based planning strategy, which converts a motion-planning problem to an online graph-search problem. In this strategy the state space is suitably discretized and the overall motion is generated by utilizing precomputed motion primitives, which connect the initial and the final states of the system, through the discretized states while avoiding obstacles during runtime. However, the lattice-based motion planners are not inherently robust to model uncertainties. The approach in \cite{gonzalez} proposes the use of graduated fidelity lattices to handle uncertainties occurring due to additive noise. The efficiency of the planner in \cite{gonzalez} is dependent on heuristics and the performance is validated through simulation and real-time experiments. The proposed robust lattice-based motion planner in this paper analytically guarantees collision free motion for uncertain nonlinear systems, which is further validated through a suitable simulation example. In this proposed approach, the lattice-based motion planner utilizes motion primitives (generated utilizing the nominal nonlinear model) along with a suitably designed feedback controller, that generates a fixed-size tube around each generated primitive. These tubes are guaranteed to contain the trajectory of the actual uncertain system affected by a bounded additive disturbance. During online implementation, a graph-search is done to connect the initial state to a region around the final state (the size of the region being dependent upon the size of the tube) utilizing the library of computed motion primitives, such that the tubes (around each utilized primitive) do not collide with any obstacle or violate any saturation constraint.\par
The contribution in this paper is twofold: firstly, a robust lattice-based motion planner is proposed for a general class of nonlinear systems, affected by bounded additive disturbance. A constrained optimal control problem (COCP) is proposed, which represents the motion-planning problem for the nominal nonlinear system. The constraints in the COCP are tightened utilizing a bounded set, containing the error between the nominal and the actual uncertain systems for all time, by virtue of a suitably designed feedback controller. This set constitutes the tube around the motion primitives. Secondly, the proposed strategy is implemented on a general Euler-Lagrange (EL) system subjected to constraints. A feedback controller for the error dynamics (between the nominal and the uncertain EL systems) as well as the bounded set (defining the tube), containing the error for all time, are explicitly designed. The performance of the proposed design is validated through a simulation experiment, which considers a ship model (satisfying the EL dynamics) subjected to path and operational constraints and affected by significant wind disturbances. It is guaranteed that the uncertain nonlinear system's trajectory stays within the tube around the planned nominal system's trajectory for all time and terminates within a region (dependent upon the size of the tube) around the final state. \par
\vspace{0.1in}
\textbf{Notations:} Given two sets \begin{small}$\mathcal{P},\mathcal{Q}$\end{small} then \begin{small}$\mathcal{P}\oplus\mathcal{Q}\triangleq \big \{a+b:a\in \mathcal{P}, b\in \mathcal{Q}\big \}$, $\mathcal{P}\ominus \mathcal{Q}\triangleq \big \{a: a\oplus \mathcal{Q}\subseteq \mathcal{P} \big \}$\end{small}. $\mathbb{R}$ and $\mathbb{N}$ are set of all real and integer numbers, respectively and $\mathbb{N}_{[m:n]}\triangleq \{m,m+1,\cdots,n-1,n\}, \forall m,n\in \mathbb{N}$ and $m<n$. The notations $\|\cdot\|$ and $\|\cdot\|_{\infty}$ denotes the $2$-norm and the infinity norm of the argument vectors/matrices, respectively. Given a vector $x\in \mathbb{R}^n$ and matrix $Q\in \mathbb{R}^{n\times n}$, the matrix norm $\|x\|_Q$ is defined as $\|x\|_Q\triangleq \|Q^{\frac{1}{2}}x\|$. The notation $diag(n_1,n_2,\cdots,n_i)$, for some finite value of $i$ and $\forall (n_j,j)\in \mathbb{R} \times \mathbb{N}_{[1:i]}$, denotes a diagonal matrix in $\mathbb{R}^{i\times i}$, with $n_j\in \mathbb{R}, \forall j \in \mathbb{N}_{[1:i]}$ as the diagonal entries of the matrix.
\section{Problem Formulation}
The paper focuses on the motion-planning problem for uncertain nonlinear systems of the form 
\begin{equation}\label{eq:nonlinear uncertain dynamics}
\dot{x}(t)=f(x(t),u(t))+d(t), \quad x(0)=x_{ini}
\end{equation}
\noindent where $x(t)$, $u(t)$ are the system state and the control input vectors, respectively. The state and control inputs are subjected to the following safety constraints
\begin{equation}
x(t)\in \mathcal{X}\backslash \mathcal{O}\triangleq \overline{\mathcal{X}}\subset \mathbb{R}^n;\quad u(t)\in \mathcal{U}\subset \mathbb{R}^m \nonumber
\end{equation}
\noindent 
where the region $\mathcal{O}\subset \mathcal{X}\subseteq \mathbb{R}^n$ represents the obstacle region in the state space, which must be avoided by the system \eqref{eq:nonlinear uncertain dynamics}. The model uncertainty is represented through $d(t)$, which accounts for uncertain exogenous disturbances affecting the system. 
\begin{assum}\label{as:disturbance}
The disturbance $d(t)$ satisfies $d(t)\in \mathcal{D}\subset \mathbb{R}^n, \forall t \in \mathbb{R}$, where $\mathcal{D}$ is a known compact set. 
\end{assum}
\noindent The concerned motion-planning problem can be formulated as a COCP as follows
\begin{subequations}\label{eq:path planning COCP}
\begin{align}
\min_{x(t),u(t),T_f} &\ \ J=\int_{0}^{T_f}l(x(t),u(t))dt \nonumber \\
& x(0)=x_{ini}; \quad x(T_f)=x_f \\
& \dot{x}(t)=f(x(t),u(t))+d(t)\label{eq:dynamic constraint} \\
& x(t)\in \overline{\mathcal{X}};\quad u(t)\in \mathcal{U}
\end{align}
\end{subequations} 
The running cost $l(x,u)$ is chosen to define the performance measure $J$. The disturbance $d(t)$ in \eqref{eq:dynamic constraint} is considered to satisfy \textit{Assumption} \ref{as:disturbance}. The COCP in \eqref{eq:path planning COCP} is designed to return a feasible trajectory for the system in \eqref{eq:nonlinear uncertain dynamics} to travel from the initial position $x_{ini}$ to a desired position $x_f$ while respecting all the imposed safety constraints. However, COCP \eqref{eq:path planning COCP} is ill-posed since the objective function neglects the presence of the disturbance term $d(t)$ and thus COCP \eqref{eq:path planning COCP} only illustrates a conceptually defined problem. This motivates the objective of this work, which is to reformulate the COCP \eqref{eq:path planning COCP} to account for the effect of model uncertainty while ensuring constraint satisfaction and to design a robust lattice-based motion planner, which will provide a feasible solution for the COCP.

\section{Robust Lattice-based Motion Planner}
This section presents a COCP that is associated with the path planning problem for the uncertain system \eqref{eq:nonlinear uncertain dynamics} and a robust lattice-based path planner to find a feasible solution for the posed COCP.
\subsection{Reachable region under the effect of model uncertainty}
The nominal disturbance-free model associated with \eqref{eq:nonlinear uncertain dynamics} is as follows
\begin{equation}\label{eq:nominal dynamics}
\dot{\overline{x}}=f(\overline{x}(t),\overline{u}(t)),\quad \overline{x}(0)=x_{ini}
\end{equation}\noindent
The error between the nominal state $\overline{x}$ and the measured state $x$ is  defined as $\tilde{x}\triangleq x-\overline{x}$. Similarly, the difference between the nominal input $\overline{u}$ and the actual implemented input $u$ is  defined as $\tilde{u}\triangleq u-\overline{u}$. Therefore, the error dynamics is computed as follows
\begin{align}\label{eq:error dynamics}
\dot{\tilde{x}}(t)=&\dot{x}-\dot{\overline{x}}=f(x(t),u(t))+d(t)-f(\overline{x}(t),\overline{u}(t))\nonumber \\
=&\tilde{f}(\tilde{x}(t),\tilde{u}(t),\overline{x}(t),\overline{u}(t))+d(t), \quad \tilde{x}(0)=0
\end{align}\noindent
where $\tilde{f}(\tilde{x}(t),\tilde{u}(t),\overline{x}(t),\overline{u}(t))\triangleq f(\overbrace{\tilde{x}+\overline{x}}^{x},\overbrace{\tilde{u}+\overline{u}}^{u})-f(\overline{x},\overline{u})$. The procedure to compute the nominal input $\overline{u}$ is depicted later in the subsection \ref{sec:RCMP}.
\begin{assum}\label{as:lyapunov conditions}
There exists a controller $\tilde{u}=\nu(t,\tilde{x})$ and a continuously differentiable function $V(t,\tilde{x})>0$ such that the following hold\\
\begin{itemize}
\item[1.] $\displaystyle{\alpha_1(\|\tilde{x}\|)\leqslant V(t,\tilde{x})\leqslant \alpha_2(\|\tilde{x}\|)}$\\
\item[2.] $\displaystyle{\frac{\partial V}{\partial \tilde{x}}\tilde{f}(\tilde{x}(t),\tilde{u}(t),\overline{x}(t),\nu(t,\tilde{x}))+\frac{\partial V}{\partial t}\leqslant -\alpha_3(\|\tilde{x}\|)+\beta(\overline{d})}$\\
\end{itemize}
where $\alpha_i(\cdot), i\in \{1,2,3\}$ and $\beta(\cdot)$ are positive definite functions of their respective arguments and $\|d\|\leqslant \overline{d}, \forall d \in \mathcal{D}$.
\end{assum}
\noindent The conditions in \textit{Assumption} \ref{as:lyapunov conditions} imply that the closed-loop error dynamics \eqref{eq:error dynamics} with respect to feedback controller $\nu(t,\tilde{x})$ is uniformly ultimately bounded (UUB). Therefore, there exists a time instant $T(\tilde{x}(0),\overline{d})>0$ and a region  $\mathcal{W}$ such that
\begin{subequations}\label{eq:all abount error rechable region} 
\begin{align}
& \hspace{0.08in} \mathcal{W}\triangleq \{\tilde{x}\in \mathbb{R}^n\ : \ \|\tilde{x}\|\leqslant \alpha_3^{-1}(\beta(\overline{d}))\}\label{eq:error reachable region}\\
& \left.\begin{matrix}
\displaystyle{\lim_{t\rightarrow T^{-}}h(\tilde{x}(t),\mathcal{W})}\rightarrow 0 \hspace{0.09in}\\
\tilde{x}(t)\in \mathcal{W}, \forall t\geqslant T \hspace{0.2in}
\end{matrix}\right\}\ \textnormal{if}\ \ \tilde{x}(0) \notin \mathcal{W}\\
& \hspace{0.089in}\tilde{x}(t)\in \mathcal{W},\forall t\in \mathbb{R}\ \hspace{0.4in} \textnormal{if}\ \ \tilde{x}(0) \in \mathcal{W} \label{eq:UUB asli wala}
\end{align}
\end{subequations}
\noindent where $h(\cdot,\cdot)$ returns the Hausdorff distance between the arguments. 
\begin{corollary}\label{cor:nominal gheshe chola}
Since $\tilde{x}=x-\overline{x}$ and $\tilde{u}=u-\overline{u}$, then for any trajectory $(\overline{x}(t),\overline{u}(t))$ of the nominal system \eqref{eq:nominal dynamics} and the controller $u(t)=\overline{u}(t)+\nu(t,\tilde{x}(t))$, the following is deduced from \eqref{eq:UUB asli wala}  
\begin{equation}
x(t)\in \overline{x}(t)\oplus \mathcal{W},\forall t\in \mathbb{R}\ \hspace{0.1in} \textnormal{if}\ \ x(0)\in \overline{x}(0)\oplus \mathcal{W} \nonumber
\end{equation}
\end{corollary}
\subsection{Reformulated COCP for motion-planning}\label{sec:RCMP}
The COCP \eqref{eq:path planning COCP} is reformulated to include implementable constraints which guarantee constraint satisfaction for the trajectories associated with the uncertain system \eqref{eq:nonlinear uncertain dynamics}. The motion-planning COCP is now based on the nominal system dynamics \eqref{eq:nominal dynamics} and is formulated as follows
\begin{subequations}\label{eq:path planning COCP reformulated}
\begin{align}
\min_{\overline{x}(t), \overline{u}(t),T_f}& \ \ J=\int_{0}^{T_f}l(\overline{x}(t),\overline{u}(t))dt \nonumber \\
& \dot{\overline{x}}(t)=f(\overline{x}(t),\overline{u}(t))\label{eq:dynamic constraint nominal} \\
& \overline{x}(t)\in \overline{\mathcal{X}}_w;\ \overline{u}(t)\in \mathcal{U}_w  \label{eq:tightened constraints}\\
& \overline{x}(0)=x_{ini}\in  \overline{\mathcal{X}}_w; \quad \overline{x}(T_f)=x_f\in \overline{\mathcal{X}}_w
\end{align}
\end{subequations}
where $\overline{\mathcal{X}}_w$ and ${\mathcal{U}}_w$ are the tightened state and input constraints, respectively, defined as follows
\begin{equation}\label{eq:tightened constraints def}
\overline{\mathcal{X}}_w \triangleq \overline{\mathcal{X}}\ominus \mathcal{W}; \  \mathcal{U}_w \triangleq \mathcal{U}\ominus \nu(t,\mathcal{W})
\end{equation} 
The state and input constraints in the reformulated COCP are tightened to guarantee overall constraint satisfaction by the states and control input of the uncertain plant \eqref{eq:nonlinear uncertain dynamics}. The magnitude of constraint tightening is dependent upon the reachable region $\mathcal{W}$ of the error dynamics \eqref{eq:error dynamics}, as defined in \eqref{eq:all abount error rechable region}.
The COCP \eqref{eq:path planning COCP reformulated} generates a feasible trajectory for the nominal system \eqref{eq:nominal dynamics}, such that it travels from the initial position $x_{ini}$ to the desired final position $x_f$, while satisfying the tightened constraints \eqref{eq:tightened constraints}. Since $x(0)=\overline{x}(0)=x_{ini}$, the following is inferred from \eqref{eq:tightened constraints} using \eqref{eq:tightened constraints def} and \textit{Corollary} \ref{cor:nominal gheshe chola}   
\begin{align}
& x(t)\in \overline{x}(t)\oplus \mathcal{W}\subset \overline{\mathcal{X}};\quad x(T_f)\in x_f\oplus \mathcal{W}\subset \overline{\mathcal{X}} \nonumber\\
& u(t)=\overline{u}(t)+\nu(t,\tilde{x}(t))\in \overline{u}(t)\oplus \nu(t,\mathcal{W})\subset \mathcal{U}\nonumber  
\end{align}
\subsection{Lattice-based motion planner with robust constraint satisfaction}
The lattice-based motion-planning strategy converts a motion-planning COCP into a discrete graph-search problem by limiting the controls to discrete subsets of available actions, represented using a set of motion primitives. At the outset, the obstacle-free feasible state space is discretized as per a desired discretization. The discretized state space consists of all reachable states, that form the graph. Subsequently, motion primitives are computed, which are feasible state and control trajectories of the concerned system connecting one reachable state with another in the discretized state space.\par
To develop a lattice-based planner for the motion-planning problem in COCP \eqref{eq:path planning COCP reformulated}, a discretized state space $\mathcal{X}_d$ is obtained from the obstacle-free tightened feasible state space $\mathcal{X}_w$ for the nominal dynamics \eqref{eq:nominal dynamics}, defined as
\begin{equation}
\mathcal{X}_w \triangleq \mathcal{X}\ominus \mathcal{W} \nonumber
\end{equation} 
A strategy to obtain the discrete state space $\mathcal{X}_d$ for lattice-based motion-planning problem can be found in \cite{bergman}. A set of motion primitives $\mathcal{M}$, associated with the nominal system \eqref{eq:nominal dynamics}, is then constructed and a motion primitive $m\in \mathcal{M}$ is defined as follows
\begin{equation}
m=(\overline{x}(t),\overline{u}(t))\in \mathcal{X}_w\times \mathcal{U}_w,\quad t\in [0,T]\nonumber 
\end{equation}
A motion primitive $m\in \mathcal{M}$ represents a feasible trajectory which moves the system \eqref{eq:nominal dynamics} from an initial state $\overline{x}(0)\in \mathcal{X}_d$ to a final state $\overline{x}(T)\in \mathcal{X}_d$, by applying the control $\overline{u}(\cdot)\in \mathcal{U}_w$. The motion-planning COCP \eqref{eq:path planning COCP reformulated} can now be approximated by the following graph-search problem, posed as a discrete COCP, which is solved online
\begin{subequations}\label{eq:LP}
\begin{align}
\min_{\{m_k\}_{k=0}^{M-1},M}& \ \ J_m=\sum_{k=0}^{M-1}l_m(\overline{x}_k,m_k) \nonumber \\
& \overline{x}_0=x_{ini}; \quad \overline{x}_M=x_f\in \overline{X}_w \\
& \overline{x}_{k+1}=f_m(\overline{x}_k,m_k) \label{eq:LP state transition}\\
& m_k\in \mathcal{M}(\overline{x}_k) \label{eq:LP motion primitive}\\
& c(\overline{x}_k,m_k)\in \overline{X}_w \label{eq:LP feasibility constraint}
\end{align}
\end{subequations}
where the decision variables are the sequence of motion primitives $\{m_k\}_{k=0}^{M-1}$ and the number of phases $M$, which is required to move from the initial position to the final position. The state transition constraint \eqref{eq:LP state transition} governs the transition of the nominal system states as follows
\begin{subequations}\label{eq:state trasitin dynamics}
\begin{align}
& \overline{x}_k=\overline{x}(t),\quad \overline{x}_{k+1}=\overline{x}(t+T)\\
& \overline{x}_{k+1}=f_m(\overline{x}_k,m_k)=\overline{x}_{k}+\int_0^Tf(\overline{x}_{\bar{t}},\overline{u}_{\bar{t}})d\bar{t}
\end{align} 
\end{subequations}
The constraint \eqref{eq:LP motion primitive} dictates the state transition from any state $\overline{x}_k$, through the available motion primitives, which encode the dynamics \eqref{eq:state trasitin dynamics}. The feasibility of the state transitions is ensured through the constraint \eqref{eq:LP feasibility constraint}, which guarantees collision free motion for the system \eqref{eq:nominal dynamics} in $\overline{\mathcal{X}}_w$.\par
The motion produced by the lattice-based planner in \eqref{eq:LP} governs the overall motion of the uncertain system \eqref{eq:nonlinear uncertain dynamics}. The following corollary proves that the uncertain system states always remain within a region $\mathcal{W}$ around the planned motion, obtained from \eqref{eq:LP} by combining a finite number of pre-computed motion primitives.
\begin{propos}
Let $x_k=x(t)$ and ${x}_{k+1}=x(t+T)$. If $x(0)=x_{ini}$, then the state $x(t)$ of the uncertain system  \eqref{eq:nonlinear uncertain dynamics} satisfies the following for all $k\in \mathbb{N}_{[1:M]}$
\begin{align}
& x_k\in \overline{x}_k\oplus \mathcal{W}\nonumber \\
& x(\bar{t})=\overline{x}(\bar{t})\oplus \mathcal{W}, \forall \bar{t} \in (t+(k-1)T,t+kT)\nonumber
\end{align}
with the control input $u(\bar{t})=\overline{u}(\bar{t})+\nu(\bar{t},\tilde{x}(\bar{t})), \forall \bar{t} \in [t+(k-1)T,t+kT]$ , where the utilized nominal control $\overline{u}(\bar{t})$ is encoded in $m_k,\forall k\in \mathbb{N}_{[1:M-1]}$.
\end{propos}
\begin{proof}
The following proof is done by the method of induction. Let 
\begin{equation}\label{eq:induction eq 1}
x_k\in \overline{x}_k\oplus \mathcal{W}
\end{equation}
Then the following is concluded from \textit{Corollary} \ref{cor:nominal gheshe chola} with the control $u(\bar{t})=\overline{u}(\bar{t})+\nu(\bar{t},\tilde{x}(\bar{t}))$, where $\overline{u}(t)$ is the control action utilized in the active motion primitive $m_k$ 
\begin{subequations}\label{eq:induction eq 2}
\begin{align}
& x(\bar{t})=\overline{x}(\bar{t})\oplus \mathcal{W}, \forall \bar{t} \in (t,t+T)\\
&x_{k+1}=x(t+T)\in \overline{x}(t+T)\oplus \mathcal{W}=\overline{x}_{k+1}\oplus \mathcal{W}
\end{align}
\end{subequations}
Since, the initial condition of the uncertain system's state satisfy  $x_0=\overline{x}_0=x_{ini}$ and the origin is an interior point of $\mathcal{W}$ (from \eqref{eq:error reachable region}), by recursively utilizing \eqref{eq:induction eq 1}-\eqref{eq:induction eq 2}, it is proved that the claimed assertions hold.  
\end{proof}
\begin{algorithm}
\vspace{0in}
\caption {\small Robust lattice-based motion planner}
\begin{algorithmic}[]
\renewcommand{\algorithmicrequire}{\textbf{Offline:}}
\REQUIRE 
\end{algorithmic}
\vspace{-0.07in}
\begin{itemize}
\item Specify $\mathcal{X},\mathcal{U}$ and $f(\overline{x},\overline{u})$.
\item Design the controller $\nu(t,\tilde{x})$ and compute $\mathcal{W}$.
\item Obtain the tightened spaces $\mathcal{X}_w=\mathcal{X}\ominus \mathcal{W}$ and  $\mathcal{U}_w=\mathcal{U}\ominus \nu(t,\mathcal{W})$.
\item Discretize $\mathcal{X}_w$ to obtain $\mathcal{X}_d$.
\item Compute the set of motion primitives $\mathcal{M}$.
\end{itemize}
\vspace{-0.05in}
\begin{algorithmic}[1]
\renewcommand{\algorithmicensure}{\textbf{Online:}}
\ENSURE \hspace{-0.02in}
\STATE Take inputs $x_{ini}$, $x_f$ and $\mathcal{O}$
\STATE Solve the discrete COCP \eqref{eq:LP}
\STATE Measure $x(t)$ and compute $\tilde{x}(t)$
\STATE Apply the control $u(t)=\overline{u}(t)+\nu(t,\tilde{x})$
\end{algorithmic}
\end{algorithm}
\section{Case study: Euler-Lagrange (EL) system}
In this section, a lattice-based planner is designed for motion-planning of uncertain EL systems of the form
\begin{subequations}\label{eq:EL system}
\begin{align}
& M\dot{q}+V_m(q)q+F(q)q+G(x)=\tau+d\\
& \dot{x}=J(x)q
\end{align}
\end{subequations}
where $q(t)\in \mathbb{R}^n$ and $x(t)\in \mathbb{R}^n$ are the generalized velocity and position vectors of the system \eqref{eq:EL system}, respectively and $\tau(t)\in \mathbb{R}^n$ is the generalized torque applied to the system. $M$ is the inertial matrix, $V_m(q)$ is the centripetal-Coriolis matrix, $F(q)$ is the friction matrix, $G(x)$ is the gravity matrix and $J(x)$ is the rotation matrix. The system is affected by an exogenous disturbance $d(t)$, which satisfies \textit{Assumption} \ref{as:disturbance} and is subjected to the following constraints
\begin{equation}
x\in \mathcal{X};\quad \dot{x}\in \dot{\mathcal{X}}; \quad \tau \in \mathcal{U} \nonumber
\end{equation} 
\begin{assum}\label{as:on M and J}
The matrices $M$ and $J(x)$ satisfy the following  for all $q\in \mathbb{R}^n$
\begin{itemize}
\item[1.] $M$ and $J(x)$ are square-invertible matrices
\item[2.] $J(x)$ is bounded as $\|J(x)\|\leqslant \mu_j$
\end{itemize}
\end{assum}
\noindent The system in \eqref{eq:EL system} is reformulated following the approach in [\textit{Section} 7.5.1, \cite{fossen}] as
\begin{equation}\label{eq:EL system reformulated neww}
M^*(x)\ddot{x}+V_m^*(x,\dot{x})\dot{x}+F^*(x,\dot{x})\dot{x}+G^*(x)=(J^{-1}(x))^T(\tau+d)
\end{equation}
where \begin{small}$M^*(x)\triangleq J^{-1}(x)^TMJ^{-1}(x)$, $V_m^*(x,\dot{x})\triangleq (J^{-1}(x)^TV_m(x,J^{-1}(x)x)-J^{-1}(x)^T\dot{J}(x)J^{-1}(x))J^{-1}(x)$, $F^*(x,\dot{x})\triangleq J^{-1}(x)^TF(J^{-1}(x)\dot{x})J^{-1}(x)$\end{small} and \begin{small}$G^*(x)\triangleq J^{-1}(x)^TG(x)$\end{small}.
The dynamics in \eqref{eq:EL system reformulated neww} is further reformulated as follows
\begin{equation}\label{eq:EL system reformulated new}
\ddot{x}=\Phi(x,\dot{x})+\Theta(x) (\tau+d)
\end{equation}
where 
\begin{subequations}\label{eq:phi theta definitions}
\begin{align}
& \Phi(x,\dot{x})\triangleq -{M^*}^{-1}(x)\Big (V_m^*(x,\dot{x})\dot{x}+F^*(x,\dot{x})\dot{x}+G^*(x)\Big ) \\
& \Theta (x)\triangleq {M^*}^{-1}(x)J^{-1}(x)^T
\end{align}
\end{subequations}
The input $\tau$ is designed as 
\begin{equation}\label{eq:tau transformation}
\tau(t)=\Theta^{-1}(x)v(t)
\end{equation}
where $v(t)$ is the control input to be formulated and applied to the system \eqref{eq:EL system} through the transformation \eqref{eq:tau transformation}. 
\begin{assum}\label{ass:control constraint reformulation}
There exists a set $\mathcal{V}$ such that for all $v(t)\in \mathcal{V}$, $\mu_{\Theta}v(t) \in \mathcal{U}$, where $\mu_{\Theta}\triangleq \max_{x\in \mathcal{X}}\|\Theta^{-1}(x)\|$
\end{assum}
Utilizing \eqref{eq:tau transformation} in \eqref{eq:EL system reformulated new}, the following is obtained
\begin{equation}\label{eq:EL final}
\ddot{x}=\Phi(x,\dot{x})+v+\Theta(x)d
\end{equation}
\begin{remark}\label{rem:on theta bound}
Since $J(x)$ is bounded (\textit{Assumption} \ref{as:on M and J}), the matrix $\Theta (x)$ is also bounded by some known scalar $\mu_{\theta}$ as $\|\Theta (x)\|\leqslant \mu_{\theta}, \forall x\in \mathbb{R}^n$.
\end{remark}
\noindent The nominal dynamics associated with \eqref{eq:EL system} is given as
\begin{subequations}\label{eq:EL nominal system}
\begin{align}
& M\dot{\overline{q}}+V_m({\overline{q}}){\overline{q}}+F({\overline{q}})\overline{q}+G(\overline{x})=\overline{\tau}\\
& \dot{\overline{x}}=J(\overline{x})\overline{q}
\end{align}
\end{subequations}
The dynamics \eqref{eq:EL nominal system} is reformulated following similar steps in \eqref{eq:EL system reformulated neww}-\eqref{eq:EL system reformulated new}, to obtain
\begin{equation}\label{eq:EL nominal dynamics reformulated new}
\ddot{\overline{x}}=\Phi(\overline{x},\dot{\overline{x}})+\Theta(\overline{x}) \overline{\tau}
\end{equation}  
where $\Phi(\cdot,\cdot)$ and $\Theta(\cdot)$ are defined in \eqref{eq:phi theta definitions}. The control input to the system \eqref{eq:EL nominal dynamics reformulated new} is also transformed following
\begin{equation}\label{eq:tau nominal transformation}
\overline{\tau}(t)=\Theta^{-1}(\overline{x})\overline{v}(t)
\end{equation}  
The overall nominal dynamics with the reformulated input in \eqref{eq:tau nominal transformation} is given as
\begin{equation}\label{eq:EL nominal final}
\ddot{\overline{x}}=\Phi(\overline{x},\dot{\overline{x}})+\overline{v}
\end{equation}
The error dynamics associated with the uncertain system \eqref{eq:EL final} and the nominal dynamics \eqref{eq:EL nominal final} is given as follows
\begin{equation}
\ddot{\tilde{x}}=\tilde{\Phi}(\tilde{x},\dot{\tilde{x}},\overline{x},\dot{\overline{x}})+\tilde{v}+\Theta(x)d\nonumber
\end{equation}
where $\tilde{x}\triangleq x-\overline{x}$, $\tilde{v}\triangleq v-\overline{v}$ and $\tilde{\Phi}(\tilde{x},\dot{\tilde{x}},\overline{x},\dot{\overline{x}})\triangleq \Phi (x,\dot{x})-\Phi(\overline{x},\dot{\overline{x}})$. The objective is to find the region $\mathcal{W}$ (defined in \textit{Corollary} \ref{cor:nominal gheshe chola}) for the states $[\tilde{x}^T,\dot{\tilde{x}}^T]^T$ by designing a suitable control $\tilde{v}$.
\subsection{Formulating robust lattice-based planner for uncertain EL system}
A feedback control input $\tilde{v}$ is chosen as follows
\begin{equation}\label{eq:feedback tilda v}
\tilde{v}(\tilde{x},\dot{\tilde{x}})=-\tilde{\Phi}(\tilde{x},\dot{\tilde{x}},\overline{x},\dot{\overline{x}})-k_1k_2\tilde{x}-(k_1+k_2)\dot{\tilde{x}}
\end{equation}
where $k_1>0$ and $k_2>0$ are chosen by the designer. The closed-loop error dynamics with the feedback control $\tilde{v}$ in \eqref{eq:feedback tilda v} is given as
\begin{equation}\label{eq:closed-loop error dynamics}
\ddot{\tilde{x}}=-k_1k_2\tilde{x}-(k_1+k_2)\dot{\tilde{x}}+\Theta(x)d
\end{equation}
\begin{theorem}\label{thm:EL stability analysis}
The closed-loop system \eqref{eq:closed-loop error dynamics} is UUB with known ultimate bound.
\end{theorem}
\begin{proof}
To analyse the closed-loop stability of \eqref{eq:closed-loop error dynamics}, a filtered tracking error $r$ associated with the tracking error $\tilde{x}$ is defined as follows
\begin{equation}\label{eq:filterd regressor}
r=\dot{\tilde{x}}+k_1\tilde{x} 
\end{equation} 
A Lyapunov function candidate $V(\tilde{x})$ is defined as follows
\begin{equation}
V(\tilde{x})= r^Tr+{\Gamma} \tilde{x}^T\tilde{x}\nonumber
\end{equation}
where $\Gamma>0$ is a scalar. The derivative of $V(\tilde{x})$ is computed as follows 
\begin{align}
&\dot{V}(\tilde{x})= 2\big (r^T\dot{r}+\Gamma\tilde{x}^T\dot{\tilde{x}}\big )\nonumber \\
= & 2\big ((\dot{\tilde{x}}+k_1\tilde{x})^T(\ddot{\tilde{x}}+k_1\dot{\tilde{x}})+\Gamma\tilde{x}^T(r-k_1\tilde{x})\big ) \nonumber \\
= & 2\big ((\dot{\tilde{x}}+k_1\tilde{x})^T(-k_2(\dot{\tilde{x}}+k_1\tilde{x})+\Theta(x)d)+\Gamma\tilde{x}^Tr-\Gamma k_1\tilde{x}^T\tilde{x}\big )\nonumber \\
= & 2\big(-k_2r^Tr+r^T\Theta(x)d+\Gamma\tilde{x}^Tr-\Gamma k_1\tilde{x}^T\tilde{x}\big) \label{eq:lyapunov analysis contd}
\end{align}
Since $\Theta (x)$ and $d$ are bounded (from \textit{Remark} \ref{rem:on theta bound} and \textit{Assumption} \ref{as:disturbance}), the term $\Theta(x)d$ is also bounded as $\|\Theta(x)d\|\leqslant \mu_{\theta}\bar{d}=D$. Utilizing this in \eqref{eq:lyapunov analysis contd}, the following is obtained
\begin{align}
&\dot{V}(\tilde{x})\leqslant 2\big( -k_2\|r\|^2+\|r\|D-\Gamma k_1\|\tilde{x}\|^2+\Gamma \|\tilde{x}\|\|r\|\big ) \nonumber \\
= &-k_2\|r\|^2 - {\Gamma k_1}\|\tilde{x}\|^2 -\Big (\sqrt{k_2}\|r\|-\frac{D}{\sqrt{k_2}}\Big )^2+\frac{D^2}{k_2}\nonumber \\
 &-\Big(\sqrt{\frac{\Gamma}{k_1}}\|r\|-\sqrt{{\Gamma k_1}}\|\tilde{x}\| \Big)^2+\frac{\Gamma}{k_1}\|r\|^2\nonumber \\
 \leqslant &- \Big (k_2-\frac{\Gamma}{k_1}\Big ) \|r\|^2 - \Gamma k_1\|\tilde{x}\|^2 +\frac{D^2}{k_2} \label{eq:ekhhuni lagbe}
\end{align}
It is inferred from \eqref{eq:ekhhuni lagbe}, that with the following gain condition
\begin{equation}\label{eq:gain condition}
k_1k_2>\Gamma
\end{equation}
$V(\tilde{x})$ decreases if
\begin{equation}\label{eq:lyapunov analysis conclusion}
\Big (k_2-\frac{\Gamma}{k_1}\Big ) \|r\|^2 + \Gamma k_1\|\tilde{x}\|^2\geqslant \frac{D^2}{k_2}
\end{equation} 
Therefore, the closed-loop system \eqref{eq:closed-loop error dynamics} is ultimately bounded, with the ultimate bound on $\tilde{x}$ and $r$ being characterized utilizing \eqref{eq:gain condition} and \eqref{eq:lyapunov analysis conclusion} as follows
\begin{equation}\label{eq:the ultimate bound 1}
\|\tilde{x}\|\leqslant C_1D; \quad \|r\|\leqslant C_2D
\end{equation}
where $C_1\triangleq \frac{1}{\sqrt{\Gamma k_1k_2}}$ and $C_2\triangleq \sqrt{\frac{k_1}{k_1k_2^2-k_2\Gamma }}$.
Utilizing \eqref{eq:filterd regressor} and \eqref{eq:the ultimate bound 1}, the ultimate bound on $\dot{\tilde{x}}$ is computed as follows
\begin{equation}\label{eq:the ultimate bound 2}
\|\dot{\tilde{x}}\|=k_1\|\tilde{x}\|+\|r\| \leqslant C_3D
\end{equation}
where $C_3=k_1C_1+C_2$. This concludes the proof.
\end{proof}
\begin{corollary}\label{cor:EL system W set}
The ultimate bounds in \eqref{eq:the ultimate bound 1} and \eqref{eq:the ultimate bound 2} are utilized to characterize the regions $\mathcal{W}_{\tilde{x}}$ and $\mathcal{W}_{\dot{\tilde{x}}}$, such that
\begin{align}
\mathcal{W}_{\tilde{x}}=& \Big \{\tilde{x}\in \mathbb{R}^{n} \ \Big | \ \|\tilde{x}\|\leqslant C_1 D, C_1= \frac{1}{\sqrt{\Gamma k_1k_2}}, \Gamma <k_1k_2\Big \} \nonumber \\ 
\mathcal{W}_{\dot{\tilde{x}}}=& \Big \{\dot{\tilde{x}}\in \mathbb{R}^{n} \ \Big | \ \|\dot{\tilde{x}}\|\leqslant C_3 D, C_3= \frac{k_1}{\sqrt{\Gamma k_1k_2}}+\nonumber \\
&\sqrt{\frac{k_1}{k_1k_2^2-k_2\Gamma }}, \Gamma <k_1k_2\Big \} \nonumber
\end{align} 
\end{corollary}
\noindent It is evident from \textit{Theorem} \ref{thm:EL stability analysis} and \textit{Corollary} \ref{cor:EL system W set} that 
\begin{equation}
\tilde{x}(t)\in \mathcal{W}_{\tilde{x}}, \dot{\tilde{x}}(t)\in \mathcal{W}_{\dot{\tilde{x}}}, \forall t\in \mathbb{R} \ \textnormal{if} \quad \tilde{x}(0)\in \mathcal{W}_{\tilde{x}},\dot{\tilde{x}}(0)\in \mathcal{W}_{\dot{\tilde{x}}}\nonumber
\end{equation}
\begin{remark}
The bounded region for the controller $\tilde{v}(\tilde{x},\dot{\tilde{x}})$ defined in \eqref{eq:feedback tilda v} is found by utilizing its upper bound as follows
\begin{equation}\label{eq:vtilde upperbound}
\|\tilde{v}(\tilde{x}´,\dot{\tilde{x}})\|=\|\tilde{\Phi}(\tilde{x},\dot{\tilde{x}},\overline{x},\dot{\overline{x}})\|+k_1k_2\|\tilde{x}\|+(k_1+k_2)\|\dot{\tilde{x}}\|
\end{equation}
The upper bounds for $\|\tilde{x}\|$ and $\|\dot{\tilde{x}}\|$ are characterized in \eqref{eq:the ultimate bound 1} and \eqref{eq:the ultimate bound 2}, respectively. The upper bound for the term $\|\tilde{\Phi}(\tilde{x},\dot{\tilde{x}},\overline{x},\dot{\overline{x}})\|$ is computed utilizing the mean value theorem as follows
\begin{align}\label{eq:phi upperbound}
& \|\tilde{\Phi}(\tilde{x},\dot{\tilde{x}},\overline{x},\dot{\overline{x}})\|=\|\Phi (x,\dot{x})-\Phi(\overline{x},\dot{\overline{x}})\|\nonumber \\
\leqslant & \underbrace{\max_{\stackon[0.5pt]{\begin{scriptsize}$\dot{x}\in \dot{\mathcal{X}}$\end{scriptsize}}{\begin{scriptsize}$x \in \mathcal{X}$\end{scriptsize}}}\left (\left \|\frac{\partial \Phi (x,\dot{x})}{\partial x}\right \|\right )}_{g_1}\|\tilde{x}\|+\underbrace{\max_{\stackon[0.5pt]{\begin{scriptsize}$\dot{x}\in \dot{\mathcal{X}}$\end{scriptsize}}{\begin{scriptsize}$x \in \mathcal{X}$\end{scriptsize}}}\left (\left \|\frac{\partial \Phi (x,\dot{x})}{\partial \dot{x}}\right \|\right )}_{g_2}\|\dot{\tilde{x}}\|\nonumber \\
= & g_1\|\tilde{x}\|+g_2\|\dot{\tilde{x}}\|
\end{align}
Therefore, utilizing \eqref{eq:vtilde upperbound}, \eqref{eq:phi upperbound} and \textit{Corollary} \ref{cor:EL system W set}, the following is concluded
\begin{align}
&\|\tilde{v}(\tilde{x}´,\dot{\tilde{x}})\|  \leqslant (g_1+k_1k_2)\|\tilde{x}\|+(g_2+k_1+k_2)\|\dot{\tilde{x}}\|\nonumber \\
\Rightarrow & \ \tilde{v}(\tilde{x}´,\dot{\tilde{x}})\in (g_1+k_1k_2)\mathcal{W}_{\tilde{x}}\oplus (g_2+k_1+k_2)\mathcal{W}_{\dot{\tilde{x}}}\nonumber \\
& \hspace{0.45in} \triangleq \tilde{v}(\mathcal{W}_{\tilde{x}},\mathcal{W}_{\dot{\tilde{x}}})\nonumber
\end{align}
\end{remark}\noindent
The lattice-based motion planner for the uncertain system \eqref{eq:EL final} is then designed following \eqref{eq:LP} and implemented  following \textbf{Algorithm} $1$. The imposed constraints are tightened as follows
\begin{equation}
\mathcal{X}_w \triangleq \mathcal{X}\ominus \mathcal{W}_{\tilde{x}}; \dot{\mathcal{X}}_w \triangleq \dot{\mathcal{X}}\ominus \mathcal{W}_{\dot{\tilde{x}}}; \mathcal{U}_w\triangleq \mathcal{U} \ominus \mu_{\Theta}\tilde{v}(\mathcal{W}_{\tilde{x}},\mathcal{W}_{\dot{\tilde{x}}})\nonumber
\end{equation}
where $\mu_{\Theta}$ is defined in \textit{Assumption} \ref{ass:control constraint reformulation}. The motion primitives are computed utilizing the nominal dynamics \eqref{eq:EL nominal final}.  
\section{Simulation Results}
In this section, the efficacy of the proposed strategy is validated through a simulation example, which considers a ship model as follows
\begin{align}
& M\dot{q}+V_m(q)q+F(q)q=\tau+d \nonumber \\
& \dot{x}=J(x)q \nonumber
\end{align}
where the variables $q=[q_1,q_2,q_3]^T,x=[x_1,x_2,x_3]^T,\tau=[\tau_1,\tau_2,\tau_3]^T$ and the parameters $M, V_m(\cdot), F(\cdot),J(\cdot)$ have the same meaning as in \eqref{eq:EL system}. The states $[x_1,x_2]^T$ are the generalized position of the ship and $x_3$ represents the heading angle of the ship. The inputs $\tau_1$ and $\tau_2$ are the applied longitudinal and the lateral forces, respectively and $\tau_3$ is the yaw torque. The constraints on the input $\tau$, are given as follows
\begin{equation}
\|\tau_1\|\leqslant 1.9\times 10^6; \ \|\tau_2\|\leqslant 1.9\times 10^6; \ \|\tau_3\|\leqslant 3.9\times 10^7 \nonumber
\end{equation}\noindent
The wind disturbance $d(t)=[d_1(t),d_2(t),d_3(t)]^T$ affecting the ship satisfies the bound given as follows
\begin{equation}
\|d(t)\|_W\leqslant 1; W^{\frac{1}{2}}=diag\left (\frac{1}{2\times 10^5},\frac{1}{12\times 10^6},\frac{1}{16\times 10^6}\right ) \nonumber
\end{equation}
\begin{figure}[H]
\begin{center}
\centering
\vspace{-0.1in}
\includegraphics[width=3.4in,height=2.2in]{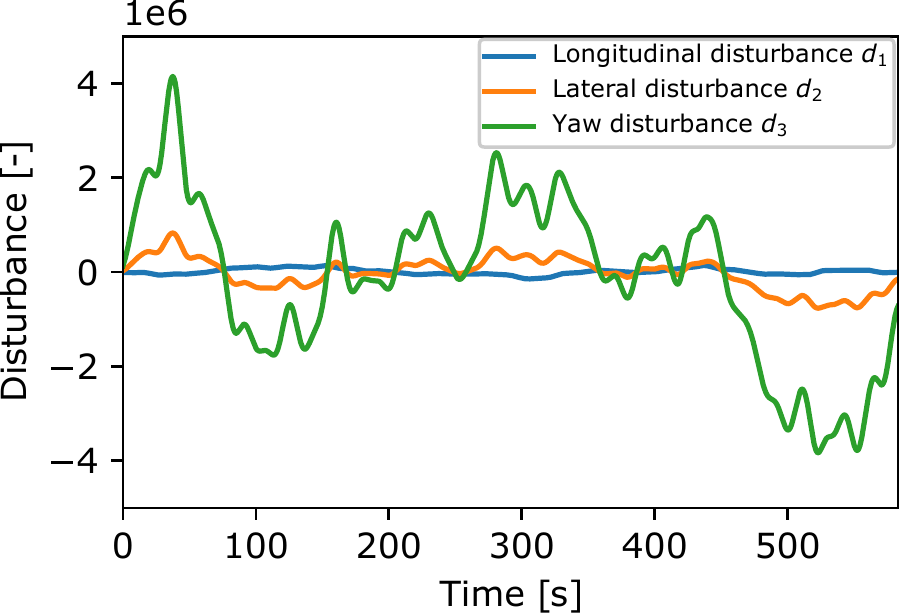}
\caption{The disturbance affecting the ship}
\label{fig:dist}
\vspace{-0.2in}
\end{center}
\end{figure} 
\noindent The longitudinal component of the wind disturbance is represented by $d_1(t)$, while $d_2(t)$ represents the lateral component. The disturbance component $d_3(t)$ represents the torque due to the wind disturbances on the body of the ship and this component has the largest magnitude among all. The nature of the disturbance $d(t)$ is illustrated in \textit{Fig.} \ref{fig:dist}. The applied feedback control $\tilde{v}(t,\tilde{x})$ has the same structure as \eqref{eq:feedback tilda v} with $k_1=k_2=0.1$. The value of $\Gamma$, that characterizes the ultimate bounds for the error $\tilde{x}$ is chosen as $\Gamma=0.009$. The plot of a subset of the motion primitives associated with the nominal disturbance-free model of the ship, with various heading angles is shown in \textit{Fig.} \ref{fig:mp}. The motion primitives are generated by solving \eqref{eq:path planning COCP reformulated} with the same discretization scheme as in \cite{bergman} using CasADi \cite{casadi} and IPOPT \cite{ipopt}. The shaded region along each of the primitives represents the tubes, which are supposed to contain the actual trajectory of the ship under the effect of the wind disturbance. The initial states of the ship are chosen as $x(0)=[50, 100, \pi /4]^T$, $q(0)=[v_{max},0,0]^T$, where $v_{max}=6$  knots and the final states are $x_f\hspace{-0.02in}=\hspace{-0.02in}[900, 900, \pi /2]^T$, $q_f\hspace{0.02in}=\hspace{0.01in}[0,0,0]^T$. \textit{Fig} \ref{fig:sm} shows the motion of the ship from the 
\begin{figure}[H]
\begin{center}
\centering
\vspace{-0.1in}
\includegraphics[width=3.3in,height=3.3in]{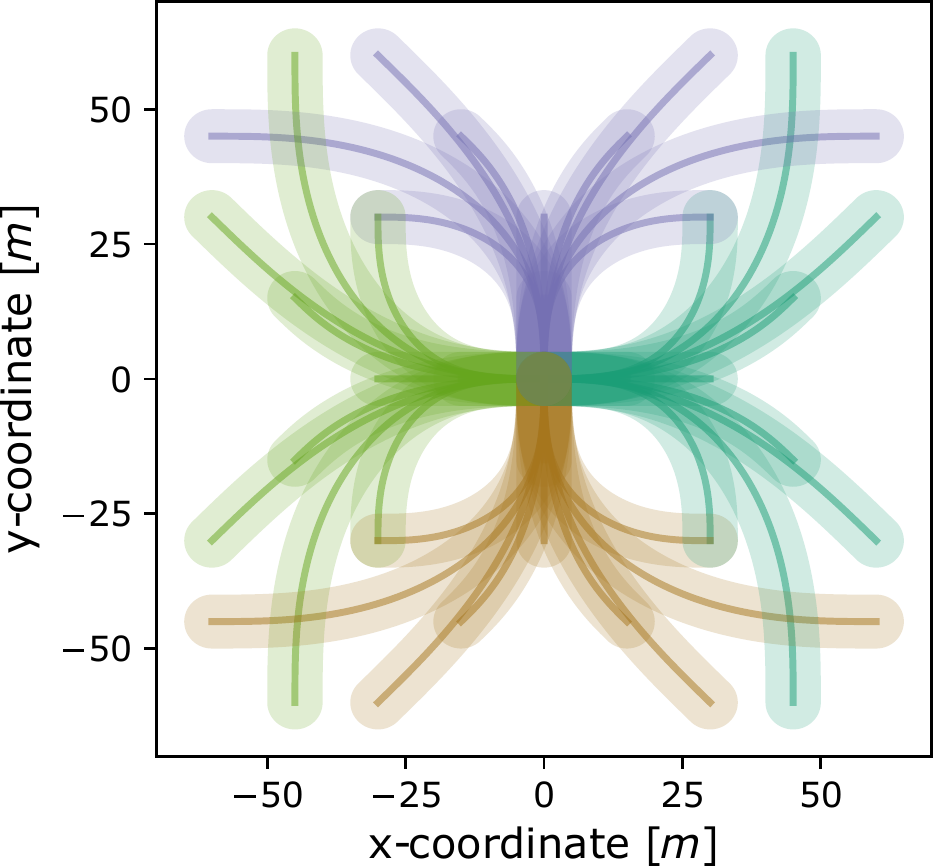}
\caption{Motion primitives with $q_1(0)\hspace{-0.03in}=\hspace{-0.03in}3$ knots and $x_3(0)\hspace{-0.02in}\in \hspace{-0.02in}\{0,\pi/2,\pi,3\pi/2\}$}
\label{fig:mp}
\vspace{-0.2in}
\end{center}
\end{figure}
\noindent chosen initial state to the final state in a environment with obstacles, for two different scenarios. \textit{Fig} \ref{fig:sm}(a) illustrates the motion of the ship under the effect of the wind disturbance with both the nominal and the feedback controller \eqref{eq:feedback tilda v} being implemented, but without considering the tubes while planning. As a result, the nominal plan is impaired from foreseeing the worst-case effect of the disturbance on the ship, due to which there is visible collision during the implemented motion, as highlighted in the zoomed box. In contrast, when the tubes are considered during motion-planning, the ship takes a detour around the obstacles to generate a collision free motion, as evident from  \textit{Fig} \ref{fig:sm}(b). The plots of the norm of error states $\hat{x}\triangleq [\tilde{x}_1,\tilde{x}_2]^T$ and the error in heading angle $\tilde{x}_3$ is shown in \textit{Fig.} \ref{fig:se}. It is seen that the errors are contained within the computed ultimate bound, which~ characterizes the tubes~ around~ the primitives. \textit{Fig.} \ref{fig:ci} illustrates the implemented control inputs along with their corresponding bounds. The nominal input is responsible for producing the nominal trajectory while the applied input, which includes the feedback controller \eqref{eq:feedback tilda v}, is responsible for keeping the overall motion of the disturbed ship within the tubes around the nominal trajectory.
\begin{figure*}[!t]
\begin{center}
\vspace{-0.15in}
\hspace{-0.85in}\begin{subfigure}{.15\textwidth}
\centering
\includegraphics[width=3.2in,height=3.2in]{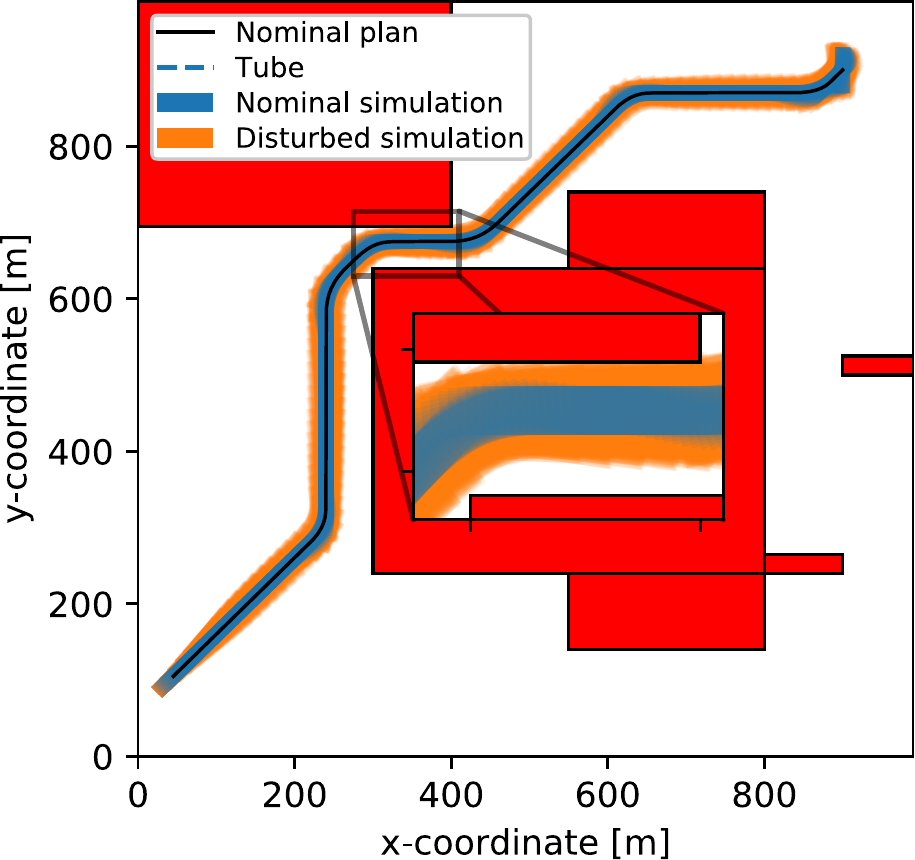}
\caption{\hspace{1.6in}(a)} 
\label{fig:law 6}
\end{subfigure}\begin{subfigure}{1.2\textwidth}
\centering
\includegraphics[width=3.2in,height=3.2in]{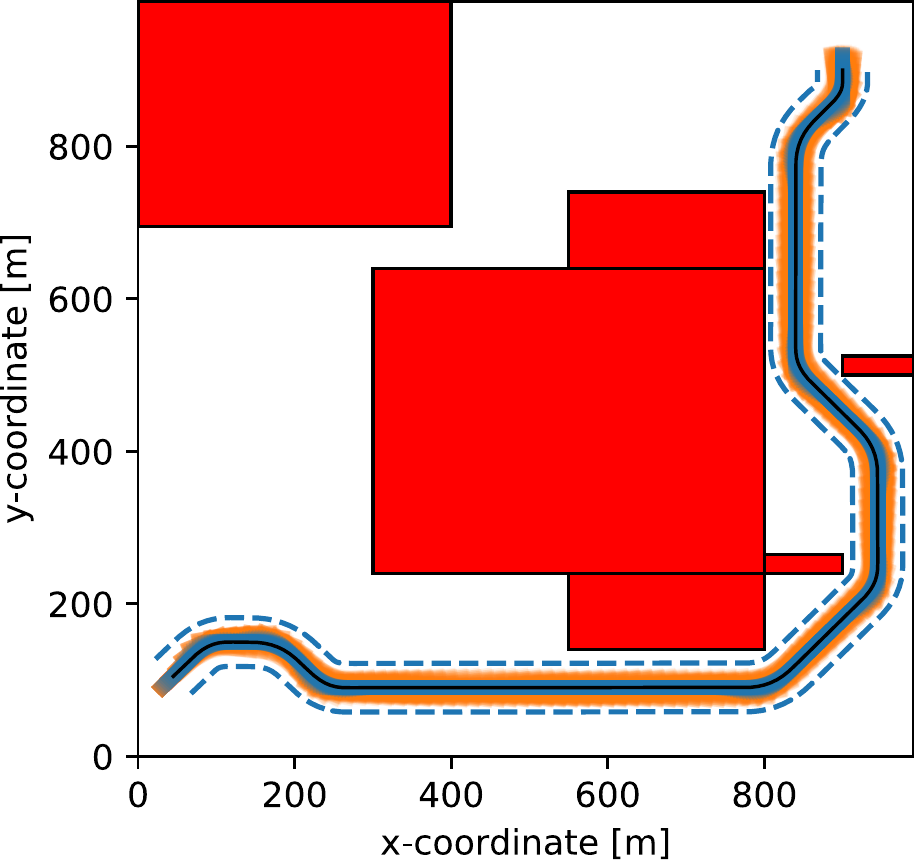}
\caption{\hspace{0.4in}(b)}  
\label{fig:law 19}
\end{subfigure}
\caption{motion-planning with wind disturbance: (a) without considering tubes (zoomed coordinates:- x-coordinate: [290m,450m], y-coordinate: [620m,730m]) (b) considering tubes}
\label{fig:sm}
\end{center}\vspace{-0.1in}
\end{figure*}
\begin{figure}[H]
\begin{center}
\centering
\vspace{-0.1in}
\includegraphics[width=3.4in,height=2.7in]{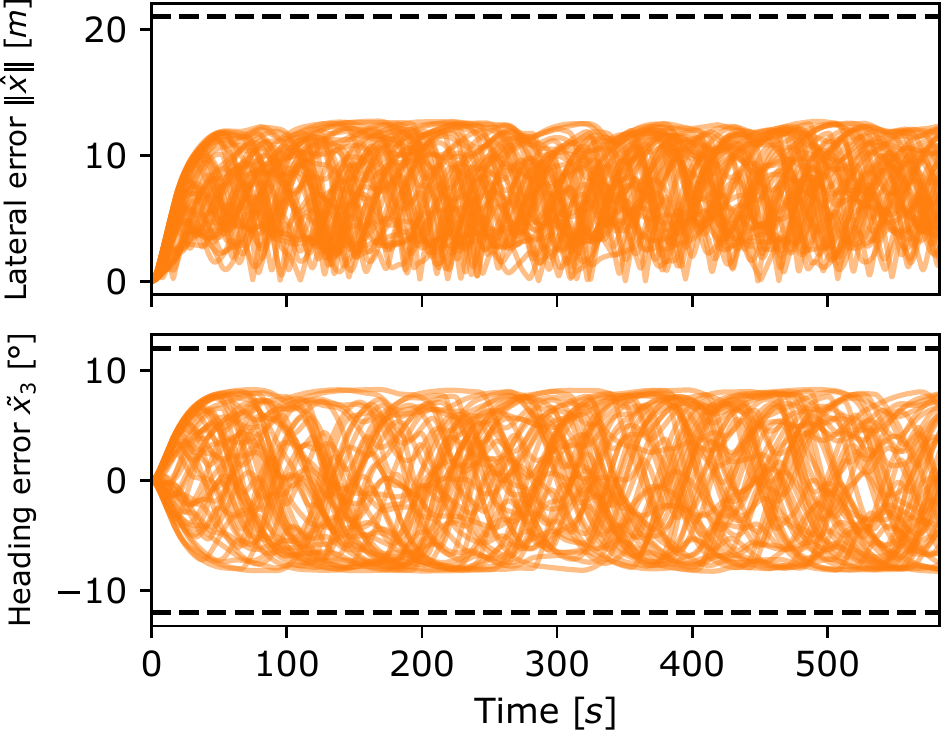}
\caption{Error in states due to mismatch between nominal and uncertain models while executing the motion in 3(b)}
\label{fig:se}
\vspace{-0.2in}
\end{center}
\end{figure}
\section{Conclusion}
A  novel robust lattice-based motion planner is proposed, which handles nonlinear systems subjected to bounded additive disturbances. The planner utilizes fixed-size tube-parametrized motion primitives, which are computed utilizing the nominal disturbance-free nonlinear system. The tube size is dictated by a suitably designed feedback controller, which keeps the error, occurring due to the mismatch between the uncertain nonlinear system and the nominal system, bounded for all time. Collision avoidance with obstacles is taken care of during runtime by solving a graph-search problem. The graph-search algorithm connects the initial state to a region around the desired final state by sequentially utilizing the tube-parameterized motion primitives, while avoiding overlap between the tubes and the obstacles. The proposed strategy is 
\begin{figure}[H]
\begin{center}
\centering
\vspace{-0.1in}
\includegraphics[width=3.4in,height=3.7in]{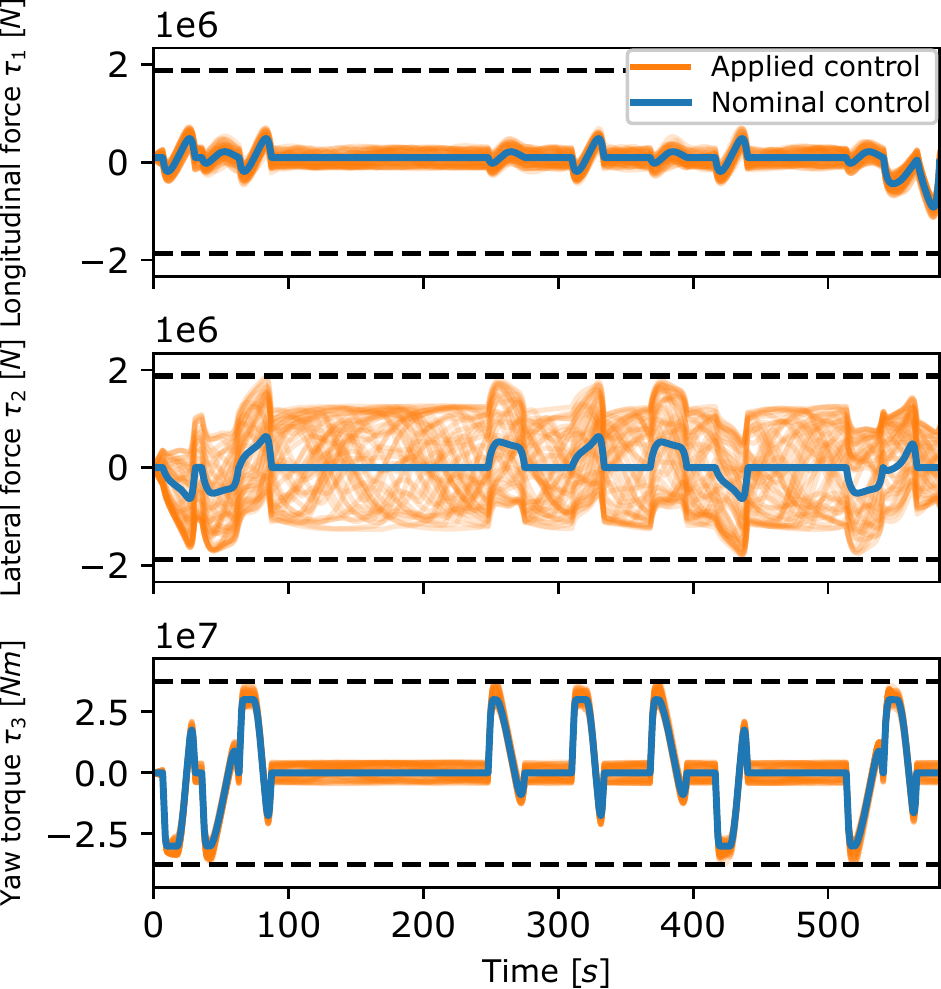}
\caption{Control inputs applied while executing the motion in 3(b)}
\label{fig:ci}
\vspace{-0.2in}
\end{center}
\end{figure}
\noindent implemented on an Euler-Lagrange (EL) system, where the feedback controller and the associated tubes are analytically derived. A ship model based on the EL dynamics is considered for simulation, where it is shown that the proposed strategy guarantees collision free motion through a fixed size tube from the initial position to the final position, while being affected by significant wind disturbance.
\bibliography{ref}
\bibliographystyle{ieeetr}
\end{document}